%% file: mpt_2.tex
\pdfoutput=1 
\RequirePackage{fix-cm}
\documentclass[reqno]{amsproc}
\usepackage{amsmath}
\usepackage{amssymb}
\usepackage{amsthm}

\usepackage{mathtools}
\mathtoolsset{showonlyrefs}

\usepackage[citation-order,nobysame]{amsrefs}
\usepackage{xyzbib}

\usepackage{microtype}
\usepackage{hyperref}

\usepackage{tikz}
\usetikzlibrary{
arrows,
decorations.pathreplacing,
decorations.markings
}
\tikzset{
->-/.style = {
very thick,
decoration = {
markings,
mark = at position 0.7 with {\arrow{latex}}
},
postaction = {decorate}
}
}

\numberwithin{equation}{section}

\newtheorem{theorem}{Theorem}
\newtheorem{lemma}{Lemma}
\newtheorem{proposition}{Proposition}

\let\cite=\cites

\let\leq=\leqslant

\newcommand{\bra}[1]{\langle #1 \rvert}
\newcommand{\ket}[1]{\lvert #1 \rangle}

\DeclareMathOperator{\qDet}{\mathrm{Det}_q}

\def\Ccal{\mathcal{C}}
\def\Hcal{\mathcal{H}}
\def\Vcal{\mathcal{V}}

\def\Cbb{\mathbb{C}}

\def\kUa{\ket{\Uparrow}}
\def\bDa{\bra{\Downarrow}}

\newcommand{\rmi}{\mathrm{i}}
\newcommand{\rme}{\mathrm{e}}
\newcommand{\wt}{\widetilde}

\begin{document}
\title{Construction of determinants
for the six-vertex model \\ with domain wall boundary conditions}

\author{Mikhail D. Minin}
\email{mmd@pdmi.ras.ru}

\author{Andrei G. Pronko}
\email{agp@pdmi.ras.ru}

\author{Vitaly O. Tarasov}
\email{vt@pdmi.ras.ru}

\address{Steklov Mathematical Institute,
Fontanka 27, St. Petersburg, 191023, Russia}


\begin{abstract}
We consider the problem of construction of determinant formulas for the
partition function of the six-vertex model with domain wall boundary conditions.
In pioneering works of Korepin and Izergin a determinant formula was proposed and proved
using a recursion relation. In later works, another determinant formulas were given
by Kostov for the rational case and by Foda and Wheeler for the trigonometric case.
Here, we develop an approach in which the recursion relation is replaced by
a system of algebraic equations with respect to one set of spectral parameters.
We prove that this system has a unique solution. The result can be easily given as a determinant
parametrized by an arbitrary basis of polynomials. In particular,
the choice of the basis of Lagrange polynomials with respect to the second set of spectral
parameters leads to the Izergin--Korepin representation, and the choice of the monomial basis leads
to the Kostov and Foda--Wheeler representations.
\end{abstract}

\maketitle
\tableofcontents

\section{Introduction}

The six-vertex model with domain wall boundary conditions (DWBC) was
introduced by Korepin in \cite{K-82} for calculation of norms of Bethe states
in the framework of the quantum inverse scattering method (QISM) 
\cite{TF-79,KBI-93}. Korepin showed that the partition function can be uniquely
determined by a list of conditions, and among those a key role is played by a
certain recursion relation. A milestone result was later obtained by Izergin
in \cite{I-87} where he obtained a determinant formula satisfying those
conditions. A detailed exposition can be found in \cite{ICK-92}.

In more recent studies, it was found that the six-vertex model with DWBC and
their generalizations, the so-called partial DWBC (pDWBC), appear in the
route of calculation of structure constants (three-point functions) in
$\mathcal{N}=4$ supersymmetric Yang--Mill theory 
\cite{EGSV-11a,EGSV-11b,GSV-12,F-12,JKKS-16}. 
In the context of these studies, and
on the grounds of Slavnov's determinant for the scalar product of Bethe
states \cite{S-89}, Kostov in \cite{K-12a,K-12b} obtained a determinant
formula for the partition function of the six-vertex model with pDWBC for the
case of rational Boltzmann weights. The Kostov's result appeared to be
intriguing, since when pDWBC are specified to DWBC it remains different from
the Izergin--Korepin representation.

Further progress was made by Foda and Wheeler in \cite{FW-12}, who argued that
the six-vertex model with pDWBC in the case of the rational weights can be
treated starting from the one with DWBC, by sending several lattice
parameters to infinity. In this procedure, Izergin--Korepin representation
leads to a determinant formula which is different from the Kostov's one, but
it turns out possible to prove that these two representations are equivalent
to each other. Foda and Wheeler have also provided a trigonometric
generalization of the Kostov's formula, albeit for a severe restriction on
the Boltzmann weights, which must be chosen with broken arrow-reversal
symmetry and obeying the stochasticity condition. As we explain in the
sequel, for DWBC these restrictions can be lifted and the formula can be used
for generic weights.

In the present paper, we propose a modified version of the original approach
of Korepin and Izergin, where we replace the Korepin's recursion relation by
a system of equations with respect to one (out of two) sets of spectral
parameters. We prove that this system of equations, supplemented by the
requirements of symmetry and polynomiality, has a unique solution that equals
the partition function of the six-vertex model with DWBC, see
Propositions \ref{prop:props-rat}\,--\,\ref{theorem:dim-trig}. On the other
hand, the solution can be easily written down as a determinant parametrized
by an arbitrary basis of polynomials. For clarity, we treat the rational
and trigonometric cases separately; our main results are stated in
Theorem~\ref{theorem:det-rat} (see Sect.~\ref{sec42}) for the rational case,
and in Theorem~\ref{theorem:det-trig} (see Sect.~\ref{sec52}),
for the trigonometric case. The proofs of the uniqueness statements,
Propositions \ref{theorem:dim-rat} and \ref{theorem:dim-trig},
follow the idea used in \cite{NPT-97}.

Choosing particular bases of polynomials, we recover previously known
representations for the partition functions. For instance, we obtain the
Izergin--Korepin representation by taking the basis of Lagrange interpolating
polynomials with respect to the second set of spectral parameters. On the
other hand, picking up the basis of monomials leads to the Kostov and
Foda--Wheeler representations.

Although we restrict ourselves here only to treating the six-vertex model with
DWBC, we believe that it can be extended to pDWBC as well. Indeed, at least
for the rational case this can be done straightforwardly, while the
trigonometric case looks subtle and requires more thorough analysis.
It is useful to mention that for the trigonometric case of the model with 
pDWBC in the limit where the lattice domain
become a semi-infinite strip shape, the partition function can be given in
terms of a Pfaffian \cite{PP-19}. Another directions to be explored concern
more general fixed boundary conditions and calculation of correlation
functions in these models \cite{R-11,PR-19,MP-21a,MP-21b,CGP-21}.

We organize our paper as follows. In the next section, we recall definition of
the six-vertex model with DWBC and list known determinant representations for
the partition function. In Sect.~\ref{sec:3}, we formulate the model within
the QISM formalism and provide various important ingredients which fix
defining properties of the partition function. In Sect.~\ref{sec:4} we study
the partition function in the rational case; our main result is the
determinant representation which depends on an arbitrary basis of
polynomials. In Sect.~\ref{sec:5} we give a generalization of this result to
the trigonometric case.

\section{The six-vertex model with DWBC}
\label{sec:2}

The six-vertex model or the ice-type model, originally is defined in terms of
arrows placed on edges of a square lattice. On the vertical edges the arrows
point up or down and on the horizontal edges the arrows point left or right.
Vertex configurations must contain equal numbers of incoming and outgoing
arrows around each lattice vertex. The six allowed vertex
configurations, in the standard order (following the conventions adopted
in \cite{B-82}), are given in Fig.~\ref{fig:six-vertices}.

\begin{figure}
\centering
\input{fig-six-vertices}
\caption{Allowed vertex configurations of the six-vertex model
and their Boltzmann weights}
\label{fig:six-vertices}
\end{figure}
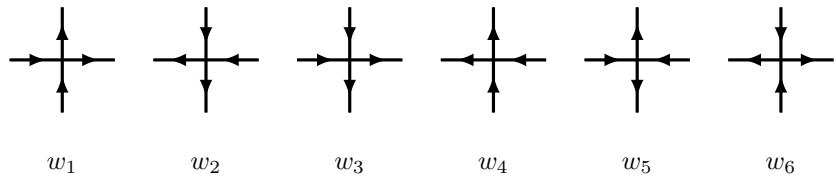

Domain wall boundary conditions (DWBC) can be defined for the six-vertex model
on a finite domain obtained by intersection of $N$ horizontal and $N$
vertical lines, the so-called $N \times N$ lattice. The DWBC mean that all
arrows on the top and bottom boundaries are incoming, on the left and right
boundaries all arrows are outgoing, see Fig.~\ref{fig:DWBC}.

To use the quantum inverse scattering method (QISM) in our study
(for a detailed exposition of the QISM, see, e.g., \cite{KBI-93}), we consider
the inhomogeneous version of the model in which the Boltzmann weights are taken
to be site-dependent in the following way. We label vertical lines from right
to left and horizontal lines from top to bottom. We assign variables
$\lambda_j$, $j=1,\ldots,N$, to the vertical lines and variables $\nu_k$,
$k=1,\ldots,N$, to the horizontal lines. The Boltzmann weight of the vertex
of the $i$th type lying at the intersection of the $j$th vertical and $k$th horizontal
lines depends only on the variables corresponding to these lines: 
$w_i=w_i(\lambda_j, \nu_k)$.

The partition function, usually denoted as $Z_N$, is defined as follows:
\begin{equation*}
Z_N=
\sum_{\Ccal \in \Omega} \prod_{j,k = 1}^{N}
w_{i(\Ccal,j,k)}(\lambda_j,\nu_k).
\end{equation*}
Here, $\Omega$ denotes the set of all allowed configurations
obeying the DWBC, and $i(\Ccal,j,k)$ is the type of the vertex
in the configuration $\Ccal$
at the intersection of the $j$th vertical and $k$th horizontal lines.

In this paper, we consider the Boltzmann weights
invariant under the reversal of all arrows,
\begin{equation*}
w_1 = w_2=:a,\qquad
w_3 = w_4=:b,\qquad
w_5 = w_6=:c.
\end{equation*}
In the rational case, the Boltzmann weights are
\begin{equation}
\label{eq:weights-rat}
a (\lambda, \nu) = \lambda - \nu + 1 ,
\qquad
b (\lambda, \nu) = \lambda - \nu,
\qquad
c = 1
\end{equation}
and in the trigonometric case are
\begin{equation}
\label{eq:weights-trig}
a (\lambda, \nu)
= \sin \gamma(\lambda - \nu+1),
\qquad
b (\lambda, \nu)
= \sin \gamma(\lambda - \nu) ,
\qquad
c = \sin\gamma.
\end{equation}

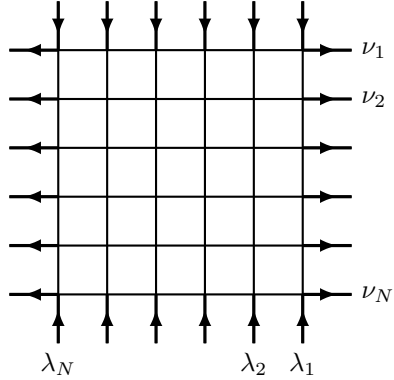
\begin{figure}
\centering
\input{fig-DWBC}
\caption{An
$N \times N$
lattice with DWBC:
all arrows on the top and bottom boundaries are incoming and
those on the left and right boundaries are outgoing}
\label{fig:DWBC}
\end{figure}

An explicit formula for the partition function
via the determinant of an $N \times N$ matrix
was derived in \cite{ICK-92} for both the rational and trigonometric cases,
\begin{multline}\label{eq:Z-IK}
Z_N =(-1)^{\frac{N(N-1)}{2}}
\frac{\prod_{j,k = 1}^N
a (\lambda_j, \nu_k) b (\lambda_j, \nu_k) }
{v_N(\{\lambda\})\, v_N(\{\nu\})}
\\ \times
\det
\left[\frac{c }{a(\lambda_j,\nu_k)b(\lambda_j,\nu_k)}\right]_{j,k=1,\ldots,N}.
\end{multline}
Here, $\{\lambda\}$ and $\{\nu\}$ stand for ordered sets
$\{\lambda\}\equiv \lambda_1,\ldots,\lambda_N$ and
$\{\nu\}\equiv \nu_1,\ldots,\nu_N$, and the factors
$v_N(\{\lambda\})$ and $v_N(\{\nu\})$ denote the ``Vandermonde determinants''
over these sets, namely, in the rational case
\begin{equation}\label{Vander-rat}
v_N(\{\lambda\})
= \prod_{1 \leq j < k \leq N} (\lambda_k - \lambda_j),
\end{equation}
and in the trigonometric case
\begin{equation}\label{Vander-trig}
v_N(\{\lambda\})
= \prod_{1 \leq j < k \leq N}\sin \gamma (\lambda_k - \lambda_j).
\end{equation}
Note that the factor $(-1)^{\frac{N(N-1)}{2}}$ in \eqref{eq:Z-IK}
can be canceled out by reversing the order of the double product in one of the two
Vandermonde factors in the denominator.

In papers \cite{K-12b,K-12a}, Kostov gave
another determinant formula for the partition function
in the rational case (see also \cite{FW-12}):
\begin{equation}\label{eq:Z-K}
Z_N = \frac{\prod_{j,k = 1}^{N}
a(\lambda_j,\nu_k) }{v_N(\{\lambda\})}
\det \left[
\lambda_j^{k - 1} - (\lambda_j + 1)^{k - 1}
\prod_{l = 1}^N\frac{b(\lambda_j,\nu_l)}{a(\lambda_j,\nu_l)}\right]_{j,k=1,\ldots,N}.
\end{equation}

A trigonometric version of formula \eqref{eq:Z-K} can be extracted from
the result of Foda and Wheeler \cite{FW-12} who proposed
a determinant formula for the six-vertex model with pDWBC with asymmetric weights restricted by
the stochastisity condition (see formula (4.19) therein and also discussion in \cite{BL-15}).
For DWBC however, one can ignore these restrictions on weights because
the numbers $n_i$ of vertices of types $i=1,\ldots,6$ in each configuration
are related by $n_1=n_2$, $n_3=n_4$, and $n_5+N=n_6$, hence the asymmetry can be removed.
The Foda--Wheeler representation specialized to DWBC and recast in our
notation yields the following formula for the partition function:
\begin{multline}\label{eq:Z-K-trig}
Z_N =\frac{\prod_{j,k=1}^{N}a(\lambda_j,\nu_k)}{(2\rmi)^{N(N-1)/2}\,v_N(\{\lambda\})}
\prod_{j=1}^{N}\rme^{-(N-2)\rmi\gamma\lambda_j-\rmi\gamma\nu_j}
\\ \times
\det \left[\rme^{2\rmi\gamma\lambda_j(k-1)}\left(1-
\rme^{\rmi \gamma (2k-N)}
\prod_{l=1}^{N} \frac{b(\lambda_j,\nu_l)}{a(\lambda_j,\nu_l)}\right)
\right]_{j,k=1,\ldots,N}.
\end{multline}
It seems that this formula is new in the context of the DWBC, and
we stress again that it is valid for arbitrary weights.
Furthermore, we also establish an analogous formula
\begin{multline}\label{eq:Z-K-trig2}
Z_N =\frac{\prod_{j,k=1}^{N}a(\lambda_j,\nu_k)}{(2\rmi)^{N(N-1)/2}\,v_N(\{\lambda\})}
\prod_{j=1}^{N}\rme^{-N\rmi\gamma\lambda_j+\rmi\gamma\nu_j}
\\ \times
\det \left[\rme^{2\rmi\gamma\lambda_j(k-1)}\left(1-
\rme^{\rmi \gamma (2k-2-N)}
\prod_{l=1}^{N} \frac{b(\lambda_j,\nu_l)}{a(\lambda_j,\nu_l)}\right)
\right]_{j,k=1,\ldots,N}.
\end{multline}

Below we derive general determinant representations of
the partition function in both rational and trigonometric cases.
The expressions above will arise under specializations of the
polynomial bases entering our representations.

\section{QISM and a basis of the Bethe vectors}
\label{sec:3}

In this section, we introduce standard objects of QISM
and formulate a proposition concerning some useful basis
of the (off-shell) Bethe vectors. We also show that there exists a family
of null Bethe vectors, which provide a system of equations for the partition function.

\subsection{QISM formulation}

We start with the consideration of the vector space $\Cbb^2$ with a standard basis
\begin{equation}
\ket{\uparrow} =
\begin{pmatrix}
1 \\ 0
\end{pmatrix},
\qquad
\ket{\downarrow} =
\begin{pmatrix}
0 \\ 1
\end{pmatrix}
\end{equation}
that we denote as spin-up and spin-down states, respectively.
The states on the edges correspond to the spin states as follows:
\begin{equation}
\begin{aligned}
\uparrow, \rightarrow
\quad \Leftrightarrow \quad
& \ket{\uparrow} \text{ or } \bra{\uparrow},
\\
\qquad \downarrow, \leftarrow
\quad \Leftrightarrow \quad
& \ket{\downarrow} \text{ or } \bra{\downarrow},
\end{aligned}
\end{equation}
where the bra-vectors are corresponding row vectors
\begin{equation}
\bra{\uparrow} = \begin{pmatrix} 1 & 0 \end{pmatrix},
\qquad
\bra{\downarrow} = \begin{pmatrix} 0 & 1 \end{pmatrix}.
\end{equation}
We associate vector space $\Cbb^2$ to each vertical
and horizontal line of the lattice.

We introduce a quantum $L$-operator which can be defined
as matrix of Boltzmann weights.
We associate operator $L_{j k} (\lambda_j, \nu_k)$ with the vertex being
at the intersection of the $j$th vertical and the $k$th horizontal lines.
This operator acts non-trivially in the direct product of two vector spaces.
The first one is the ``vertical'' space $\Vcal_j = \Cbb^2$,
associated with the $j$th vertical line,
and the second one is the ``horizontal'' space $\Hcal_k = \Cbb^2$,
associated with the $k$th horizontal line.
We regard the ``vertical'' space $\Vcal_j$ as auxiliary
and the ``horizontal'' space $\Hcal_k$ as quantum.
Graphically,
the $L$-operator acts from top to bottom and from right to left.
In other words, arrow states on the top and right edges of the vertex
correspond to the ``in'' indices of the matrix of the $L$-operator,
and arrow states on the bottom and left edges
correspond to the ``out'' indices.

The $L$-operator is given by
\begin{equation}\label{eq:L-rat}
L_{jk} (\lambda,\nu) =
\begin{pmatrix}
a(\lambda,\nu)\pi_k^{+}+b(\lambda,\nu) \pi_k^{-} & c\sigma_k^{-}
\\
c\sigma_k^{+} & b(\lambda,\nu)\pi_k^{+} + a(\lambda,\nu)\pi_k^{-}
\end{pmatrix}_{[\mathcal{V}_j]}.
\end{equation}
Here, $\sigma^\pm_k$ and $\pi^\pm_k$ denote basis operators
acting in the space $\Hcal_k$,
\begin{equation}
\sigma^+=
\begin{pmatrix} 0&1\\0&0
\end{pmatrix},\quad
\sigma^-=
\begin{pmatrix} 0&0\\1&0
\end{pmatrix},\quad
\pi^+=
\begin{pmatrix} 1&0\\0&0
\end{pmatrix},\quad
\pi^-=
\begin{pmatrix} 0&0\\0&1
\end{pmatrix},
\end{equation}
and the superscript $[\mathcal{V}_j]$ in \eqref{eq:L-rat} indicates
that this is a matrix in the space $\Vcal_j$.

Now we construct the monodromy matrix $T_j (\lambda_j)$,
multiplying $L$-operators along the $j$th vertical line,
\begin{equation}\label{Tmatrix}
T_j (\lambda_j) =
L_{j N} (\lambda_j, \nu_N)
\cdots
L_{j 2} (\lambda_j, \nu_2)
L_{j 1} (\lambda_j, \nu_1)
=
\begin{pmatrix}
A(\lambda_j) & B(\lambda_j) \\
C(\lambda_j) & D(\lambda_j) \\
\end{pmatrix}_{[\mathcal{V}_j]}.
\end{equation}
Operators $A (\lambda)$, $B (\lambda)$, $C (\lambda)$ and $D (\lambda)$
act in the total quantum space $\Hcal = \otimes_{k = 1}^{N} \Hcal_k$. Note that
these operators depend also on the parameters $\nu_1, \dots, \nu_N$,
that is $A (\lambda)=A (\lambda;\nu_1, \dots, \nu_N)$, etc. Although this
dependence is important, we omit it in the notation.

To express the partition function in the operator form, we note
that the boundary conditions on the right boundary correspond to the vacuum-state
\begin{equation}
\ket{\Uparrow} := \ket{\uparrow_N\dots\uparrow_1}
:=\ket{\uparrow_N}\otimes \cdots \otimes \ket{\uparrow_1},
\end{equation}
where $\ket{\uparrow_k}$ is the basis vector in the space $\Hcal_k$, $k=1,\ldots,N$. 
The boundary conditions on the left boundary are described by the bra-vector
\begin{equation}
\bra{\Downarrow} :=\bra{\downarrow_N\ldots \downarrow_1}.
\end{equation}
The partition function reads:
\begin{equation}
\label{eq:Z-mel}
Z_N =\bDa B (\lambda_N) \cdots B (\lambda_2) B (\lambda_1)\kUa.
\end{equation}

Notice that the ``all spins up'' state
$\ket{\Uparrow}$ is annihilated by the operator $C(\lambda)$, 
\begin{equation}\label{CUp0}
C(\lambda)\ket{\Uparrow}=0.
\end{equation}
It is an eigenstate of the operators $A (\lambda)$ and $D (\lambda)$,
\begin{equation}\label{eq:a-and-d-def}
A (\lambda) \kUa = a (\lambda) \kUa,
\qquad
D (\lambda) \kUa = d (\lambda) \kUa,
\end{equation}
with the eigenvalues
\begin{equation}\label{eq:a-and-d}
a (\lambda)= \prod_{j=1}^{N} a (\lambda, \nu_j), \qquad
d (\lambda)= \prod_{j=1}^{N} b (\lambda, \nu_j).
\end{equation}
We will use these properties later.
\subsection{The Yang-Baxter algebra}

The key relation of QISM is the
RLL-relation, which reads
\begin{equation}
R_{jk} (\lambda, \mu)
\left[L_{j l} (\lambda, \nu) \otimes L_{k l} (\mu, \nu)
\right]
=
\left[L_{k l} (\mu, \nu) \otimes L_{j l} (\lambda, \nu)\right]
R_{jk} (\lambda, \mu).
\end{equation}
Here $R_{jk} (\lambda, \mu)$ is the
$R$-matrix that acts non-trivially in the tensor product of two auxiliary spaces
$\Vcal_{j} \otimes \Vcal_{k}$ and has the following form
as a $4 \times 4$ matrix,
\begin{equation}
R_{jk} (\lambda, \mu)
=
\begin{pmatrix}
f (\lambda, \mu) & 0 & 0 & 0\\
0 & 1 & g (\lambda, \mu) & 0\\
0 & g (\lambda, \mu) & 1 & 0\\
0 & 0 & 0 & f (\lambda, \mu)
\end{pmatrix}_{[\Vcal_{j} \otimes \Vcal_{k}]}
.
\end{equation}
The functions $f (\lambda, \mu)$ and $g (\lambda, \mu)$
are given by
\begin{equation}
f (\lambda, \mu) = \frac{a(\lambda,\mu)}{b(\lambda,\mu)},
\qquad
g (\lambda, \mu) =
\frac{c}{b(\lambda,\mu)}.
\end{equation}

The RLL-relation leads to the same relation for the monodromy matrix, the
RTT-relation
\begin{equation}\label{RTT}
R_{jk} (\lambda, \mu)
\left[
T_j (\lambda) \otimes T_k (\mu)
\right]
=
\left[
T_k (\mu) \otimes T_j (\lambda)
\right]
R_{jk} (\lambda, \mu).
\end{equation}
This relation encodes commutation relations between the operators
$A (\lambda)$, $B (\lambda)$, $C (\lambda)$ and $D (\lambda)$, called
the Yang-Baxter algebra. For our purpose, we will use only some
of these commutation relations,
namely, the commutation relation of the $B$-operator with itself,
\begin{equation}\label{eq:BB}
\left[ B (\lambda), B (\mu) \right] = 0,
\end{equation}
and the following commutation relation between the $A$- and $B$-operators:
\begin{equation}\label{eq:AB}
A(\mu) B(\lambda) =
f(\lambda,\mu) B(\lambda) A(\mu) + g(\mu,\lambda) B(\mu) A(\lambda).
\end{equation}
In addition to these relations, we also use below
an explicit expression for the eigenvalue $a (\lambda)$
under the action of the operator $A(\lambda)$ on the vacuum vector.

\subsection{The quantum determinant}

The Yang--Baxter algebra \eqref{RTT} has a central element, called
the quantum determinant \cite{IK-81,K-82b}. It is defined as follows:
\begin{equation}\label{qdetdef}
\qDet T(\lambda)
:= D (\lambda) A (\lambda - 1) - C (\lambda) B (\lambda - 1).
\end{equation}
Another formula for $\qDet T(\lambda)$ is
\begin{equation}\label{qdet2}
\qDet T(\lambda)
= A (\lambda) D (\lambda - 1) - B (\lambda) C (\lambda - 1).
\end{equation}
The fact that the quantum determinant is a central element of the algebra means
that it commutes with all the entries of the monodromy matrix,
\begin{equation}
[\qDet T(\lambda), X(\mu)]=0,\qquad X\in \{A,B,C,D\}.
\end{equation}
The action of the quantum determinant on the quantum space is proportional
to the identity operator. The proportionality coefficient can be computed by acting 
with the expression \eqref{qdet2} on the state 
$\ket{\Uparrow}$ using \eqref{CUp0} and \eqref{eq:a-and-d-def}, that yields
\begin{equation}\label{qdet-ad}
\qDet T(\lambda)
= a (\lambda) d(\lambda-1).
\end{equation}
The properties of the quantum determinant follow from the commutation relations
encoded in \eqref{RTT}, for a detailed exposition see, e.g., \cite{KBI-93}, Ch.~VII.

\subsection{Eigenvectors of the $A$-operator}

Here we consider a construction of the eigenvectors of the operator $A(\lambda)$.

\begin{proposition}\label{prop:eigenA}
All the eigenvectors of the operator $A (\mu)$
are given by the following non-zero Bethe vectors
\begin{equation}\label{eq:A-eigenvector}
B ( \nu_{j_1} - 1) B ( \nu_{j_2} - 1 ) \cdots
B ( \nu_{j_k} - 1 ) \ket{\Uparrow},
\end{equation}
with the eigenvalues
\begin{equation}\label{eq:A-eigenvalue}
\Lambda (\mu ; \nu_{j_1}, \dots, \nu_{j_k}) =a(\mu)
\prod_{l = 1}^k f (\nu_{j_l} - 1, \mu),
\end{equation}
where $k \leq N$,
and $\nu_{j_l} \neq \nu_{j_m}$
for any $l, m = 1, \dots, k$,
$l \neq m$.
\end{proposition}

\begin{proof}
Let us consider the Bethe vector $B (\nu_{j_1} - 1) \kUa$.
We act on the Bethe vector with $A (\mu)$
and use commutation relation \eqref{eq:AB}
to get
\begin{multline}
A(\mu) B(\nu_{j_1}-1) \kUa =
a(\mu) f(\nu_{j_1}-1, \mu) B(\nu_{j_1}-1) \kUa
\\
+ a(\nu_{j_1}-1) g(\mu,\nu_{j_1}-1)B(\mu) \kUa.
\end{multline}
The second term equals zero because
the eigenvalue of the operator $A (\mu)$, the function $a(\mu)$,
has zeros at the points $\mu = \nu_j - 1$, $j = 1, \dots, n$,
see \eqref{eq:a-and-d}, \eqref{eq:weights-rat}, \eqref{eq:weights-trig}.
Hence, the Bethe vector $B (\nu_{j_1} - 1) \kUa$ is an eigenvector
of the operator $A (\mu)$ with the eigenvalue
\begin{equation}
\Lambda (\mu ; \nu_{j_1} )
= a (\mu) f (\nu_{j_1} - 1, \mu),
\end{equation}
unless it is the null vector.

To show that the Bethe vector is non-zero,
we consider the quantum determinant \eqref{qdetdef}.
We act with $\qDet T(\nu_{j_1})$ on the vacuum-state
\begin{multline}\label{eq:q-det-action}
\qDet T (\nu_{j_1}) \kUa = D (\nu_j) a (\nu_{j_1} - 1) \kUa
- C (\nu_{j_1}) B (\nu_{j_1} - 1) \kUa
\\
= - C (\nu_{j_1}) B (\nu_{j_1} - 1) \kUa
\end{multline}
and use the fact that $\qDet T(\nu_{j_1})$ is proportional to the identity operator
multiplied by the $a(\nu_{j_1})d(\nu_{j_1}-1)$, see \eqref{qdet-ad}.
Therefore, the right-hand side of \eqref{eq:q-det-action}
is not equal to zero, hence the Bethe vector $B (\nu_{j_1} - 1) \kUa$
does not vanish as well.

The rest of the proof is by induction on the number of $B$-operators. It
is fairly straightforward and we omit the details.

Since $\nu_{j_l} \neq \nu_{j_m}$ for $l, m = 1, \dots, N$, and $l \neq m$,
we get all $2^N$ non-zero eigenvectors of the operator $A(\lambda)$
with different eigenvalues given by \eqref{eq:A-eigenvalue}.
\end{proof}

\subsection{Null Bethe vectors}

An important application of the explicit form for the
eigenstates of the operator $A(\mu)$ is the following
fact, playing a conceptual role in what follows.

\begin{lemma}\label{prop:null}
The vector $B (\nu_j) B (\nu_j - 1) \kUa$ equals zero for each $j=1,\dots,N$.
\end{lemma}
\begin{proof}
We proof the assertion by contradiction. We act with the operator $A (\mu)$
on the vector $B (\nu_j) B (\nu_j-1) \kUa$ and use commutation relation \eqref{eq:AB} 
to get
\begin{multline}
A (\mu) B (\nu_j) B (\nu_j - 1) \kUa
=
f (\nu_j, \mu)B (\nu_j) A (\mu) B (\nu_j - 1) \kUa
\\
-g (\nu_j, \mu)B (\mu) A (\nu_j) B (\nu_j - 1) \kUa.
\end{multline}
The second term vanishes because the eigenvalue of operator $A (\nu_j)$ is equal to zero,
see \eqref{eq:A-eigenvalue}.
Hence, the vector $B (\nu_j) B (\nu_j - 1) \kUa$ is either
an eigenvector of the operator $A (\mu)$ with the eigenvalue
\begin{equation}
\label{eq:B-eigenvalue}
\Lambda (\mu; \nu_j,\nu_j-1) =a(\mu) f(\nu_j,\mu) f(\nu_j-1,\mu)
\end{equation}
or the null vector. By Proposition \ref{prop:eigenA},
we have already constructed the complete set of eigenvectors 
of the operator $A (\mu)$, and none of eigenvalues \eqref{eq:A-eigenvalue}
is equal to $\Lambda (\mu; \nu_j,\nu_j-1)$. Therefore,
the vector $B (\nu_j) B (\nu_j - 1) \kUa$ is equal to zero.
\end{proof}

\subsection{Spin flip operators}

The operators $B(\nu_j)$, $j=1,\ldots,N$, play the role of operators flipping
spins for certain vectors. More precisely, we have the following simple
but useful
property.

\begin{lemma}\label{prop:spinflip}
The operator $B(\nu_j)$ flips the $j$th spin
in the vector $\ket{\uparrow_N\ldots\uparrow_{j+1}\uparrow_j\downarrow_{j-1}\ldots\downarrow_1}$, namely
\begin{equation}\label{spinflip}
B(\nu_j)\ket{\uparrow_N\ldots\uparrow_{j+1}\uparrow_j\downarrow_{j-1}\ldots\downarrow_1}=
a(\nu_j)\ket{\uparrow_N\ldots\uparrow_{j+1}\downarrow_j\downarrow_{j-1}\ldots\downarrow_1},
\end{equation}
where $a(\lambda)$ is the vacuum eigenvalue of the operator $A(\lambda)$, given by
\eqref{eq:a-and-d}.
\end{lemma}

\begin{proof}
We represent the $B$-operator as the $(1,2)$-entry of the monodromy matrix,
$B(\lambda)=[T(\lambda)]_{1,2}$. Acting with monodromy matrix $T(\nu_j)$
using its definition \eqref{Tmatrix}
on the vector $\ket{\uparrow_N\ldots\uparrow_{j+1}\uparrow_j\downarrow_{j-1}\ldots\downarrow_1}$
and picking up the relevant entry, yields
\begin{multline}
B(\nu_j)\ket{\uparrow_N\ldots\uparrow_{j+1}\uparrow_j\downarrow_{j-1}\ldots\downarrow_1}
\\
=
\bigg[\begin{pmatrix}
a(\nu_j,\nu_N)\ket{\uparrow_N} & c\ket{\downarrow_N}\\ 0 & b(\nu_j,\nu_N)\ket{\uparrow_N}
\end{pmatrix}\cdots
\begin{pmatrix}
a(\nu_j,\nu_{j+1})\ket{\uparrow_{j+1}} & c\ket{\downarrow_{j+1}}\\ 0 & b(\nu_j,\nu_{j+1})\ket{\uparrow_{j+1}}
\end{pmatrix}
\\ \times
\begin{pmatrix}
a(\nu_j,\nu_j)\ket{\uparrow_j} & c\ket{\downarrow_j}\\ 0 & 0
\end{pmatrix}
\begin{pmatrix}
b(\nu_j,\nu_{j-1})\ket{\downarrow_{j-1}} & 0 \\ c\ket{\uparrow_{j-1}} & a(\nu_j,\nu_{j-1})\ket{\downarrow_{j-1}}
\end{pmatrix}\cdots
\\ \times
\cdots
\begin{pmatrix}
b(\nu_j,\nu_{1})\ket{\downarrow_{1}} & 0\\ c\ket{\uparrow_{1}} & a(\nu_j,\nu_{1})\ket{\downarrow_{1}}
\end{pmatrix}\bigg]_{1,2}
\\
=c \prod_{\substack{k = 1\\ k \neq j}}^{N}a(\nu_j,\nu_k)
\ket{\uparrow_N\ldots\uparrow_{j+1}\downarrow_j\downarrow_{j-1}\ldots\downarrow_1}.
\end{multline}
Since $a(\nu_j,\nu_j)=c$ and taking into account \eqref{eq:a-and-d},
we get \eqref{spinflip}.
\end{proof}

As a corollary, the operator 
$B(\nu_j)\cdots B(\nu_1)$ applied to the vector $\ket{\Uparrow}$ flips $j$ spins,
\begin{equation}
B(\nu_j)\cdots B(\nu_1)\ket{\Uparrow}= \prod_{k=1}^{j} a(\nu_k)
\ket{\uparrow_N\dots\uparrow_{j+1}\downarrow_{j}\ldots\downarrow_{1}}.
\end{equation}
In the important case $j=N$,
\begin{equation}\label{BNUp}
B(\nu_N)\cdots B(\nu_1)\ket{\Uparrow}= \prod_{k=1}^{N} a(\nu_k)
\ket{\Downarrow},
\end{equation}
that is closely related to the partition function $Z_N$, see \eqref{eq:Z-mel}.

\section{Determinant representations in the rational case}
\label{sec:4}

In this section, we state the properties
of the partition function $Z_N(\{\lambda\})$
in the rational case. We prove that the partition function
can be written as a determinant involving an arbitrary basis of polynomials.
We also show that
representations \eqref{eq:Z-IK} and \eqref{eq:Z-K} arise
under specializations to the Lagrange and monomial bases, respectively.

\subsection{Properties of the partition function}

We consider the partition function $Z_N$ defined by \eqref{eq:Z-mel}. We will
study it as a function of the variables $\lambda_1, \dots, \lambda_N$,
considering the variables $\nu_1, \dots, \nu_N$ as parameters. As before, we
use the notation $\{\lambda\}\equiv \lambda_1, \dots, \lambda_N$. We begin
with describing several properties of the function $Z_N(\{\lambda\})$ and,
then we will show that a function with these properties is unique.

\begin{proposition}\label{prop:props-rat}
The partition function $Z_N (\{\lambda\})$ 
of the six-vertex model with DWBC in the rational parametrization,
as a function of the variables $\lambda_1, \dots, \lambda_N$
has the following properties:
\begin{enumerate}
\item
It is a symmetric polynomial in $\lambda_1, \dots, \lambda_N$;
\item
The degree in each of the variables $\lambda_1, \dots, \lambda_N$
equals $N - 1$;
\item
For each pair of
variables, say $\lambda_1,\lambda_2$,
one has
\begin{equation}\label{eqnunu}
Z_N(\nu_j, \nu_j - 1, \lambda_3, \dots, \lambda_N) = 0,
\qquad j = 1, \dots, N;
\end{equation}
\item
As $\{\lambda\}=\{\nu\}$, the following holds:
\begin{equation}\label{Znunu}
Z_N(\{\nu\})
= \prod_{j,k = 1}^N (\nu_j - \nu_k + 1).
\end{equation}
\end{enumerate}
\end{proposition}
\begin{proof}
The first and second properties follow from the explicit form of the
$L$-operator, see \eqref{eq:L-rat}, and the commutativity of the
$B$-operators. The third property is due to Lemma~\ref{prop:null}. The fourth
property follows from relations \eqref{eq:Z-mel} and
\eqref{BNUp}.
\end{proof}

As a remark,
notice that one can prove property \eqref{Znunu}
without Lemma~\ref{prop:spinflip}
using
Korepin's
reduction relation \cite{K-82}.
The
starting point is the fact
that
$L$-operator \eqref{eq:L-rat}
with coinciding arguments
is just the permutation operator,
$L_{jk}(\nu,\nu)=P_{jk}$.
Consider
the vertex at the top right corner of the lattice, see Fig.~\ref{fig:DWBC}.
The DWBC mean that at $\lambda_1=\nu_1$
only the arrow configuration of type 6 can contribute,
see Fig.~\ref{fig:six-vertices}. This observation immediately
yields
the relation
\begin{multline}\label{Zreduction}
Z_N(\nu_1,\lambda_2,\ldots,\lambda_N)
=\prod_{k=2}^{N}(\nu_1-\nu_k+1)
\prod_{j=2}^{N}(\lambda_j-\nu_1+1)
\\ \times
Z_{N-1}(\lambda_2,\ldots,\lambda_N;\{\nu\}\setminus \nu_1),
\end{multline}
where
it is indicated
in the notation
that the partition function
$Z_{N-1}$ depends on the set of parameters
$\{\nu\}\setminus \nu_1\equiv\nu_2,\ldots,\nu_N$.
Repeating relation \eqref{Zreduction} iteratively
for the remaining variables
by setting
subsequently
$\lambda_2=\nu_2$, $\lambda_3=\nu_3$, $\ldots$, $\lambda_N=\nu_N$,
one
obtains
\eqref{Znunu}.

Let us now
study
the system of equations \eqref{eqnunu}, \eqref{Znunu}
and show that it has a unique solution.

\begin{proposition}\label{theorem:dim-rat}
Let
$P_N (\{\lambda\})$ be a
polynomial in
the variables
$\lambda_1, \dots, \lambda_N$,
depending on the
parameters $\{\nu\}=\nu_1,\ldots,\nu_N$,
that
has the following properties:
\begin{enumerate}
\item
It is symmetric in $\lambda_1, \dots, \lambda_N$;
\item
The degree in each of the variables $\lambda_1, \dots, \lambda_N$
equals $N - 1$;
\item
For each pair of
variables, say $\lambda_1,\lambda_2$,
one has
\begin{equation}\label{PNeq}
P_N(\nu_j, \nu_j - 1,\lambda_3, \dots, \lambda_N) =0,\qquad
j = 1, \dots, N;
\end{equation}
\item
As $\{\lambda\}=\{\nu\}$, the following holds:
\begin{equation}\label{Pnunu}
P_N(\{\nu\}) = \prod_{j,k = 1}^N (\nu_j - \nu_k + 1).
\end{equation}
\end{enumerate}
Then, $P_N (\{\lambda\})
=Z_N (\{\lambda\})$.
\end{proposition}
\begin{proof}
The proof is by induction on $N$. The base of induction at $N=1$ is clear.

For the induction step, expand the polynomial $P_N (\{\lambda\})$ in
the variable $\lambda_1$ using the Lagrange interpolating polynomials
of degree $N-1$ for the points $\nu_1,\ldots,\nu_N$,
\begin{equation}\label{PNL1}
P_N (\{\lambda\})
= \sum_{i = 1}^N
P_N (\nu_i, \lambda_2, \dots, \lambda_N)
\prod_{\substack{j = 1\\ j \neq i}}^{N}
\frac{ \lambda_1 - \nu_j }{\nu_i - \nu_j}.
\end{equation}
It follows from \eqref{PNeq} and the symmetry of the variables $\lambda_1,\dots,\lambda_N$ that
\begin{equation}\label{Pnui}
P_N (\nu_i, \lambda_2, \dots, \lambda_N)
=
S_{N-1,i} (\lambda_2, \dots, \lambda_N)
\prod_{k = 2}^N (\lambda_k - \nu_i + 1)\,\prod_{j = 1}^N (\nu_i - \nu_j + 1),
\end{equation}
$i=1,\ldots,N$, for some polynomials
$S_{N-1,i} (\lambda_2, \dots, \lambda_N)$ in $\lambda_2, \dots, \lambda_N$.
The second product in \eqref{Pnui} provides a convenient normalization
of the polynomials $S_{N-1,i}$.

It is straightforward to verify from the properties (1)\,--\,(4)
of $P_N (\{\lambda\})$ that for each $i=1,\ldots,N$, the polynomial
$S_{N-1,i} (\lambda_2, \dots, \lambda_N)$ has the following properties:
\begin{enumerate}
\item
It is symmetric in $\lambda_2, \dots, \lambda_N$;
\item
The degree in each of the variables $\lambda_2, \dots, \lambda_N$
equals $N - 2$;
\item
For each pair of the variables, say $\lambda_2,\lambda_3$,
one has
\begin{equation}\label{SNeq}
S_{N-1,i}(\nu_j, \nu_j - 1,\lambda_4, \dots, \lambda_N) =0,\qquad
j = 1, \dots, i-1, i+1,\ldots, N;
\end{equation}
\item
The following holds:
\begin{equation}\label{Snunu}
S_{N-1,i}(\{\nu\}\setminus \nu_i) =
\prod_{\substack{j,k = 1\\j,k\ne i}}^N (\nu_j - \nu_k + 1),
\end{equation}
where we use the notation
$\,\{\nu\}\setminus \nu_i:=\nu_1,\ldots,\nu_{i-1},\nu_{i+1},\ldots,\nu_N$.
\end{enumerate}
Hence by the induction assumption, for each $i=1,\ldots,N$, one has
\begin{equation}\label{SZX}
S_{N-1,i} (\lambda_2, \dots, \lambda_N)
= Z_{N-1} (\lambda_2, \dots, \lambda_N;\{\nu\}\setminus \nu_i),
\end{equation}
indicating in the right-hand side the dependence on the parameters $\{\nu\}$ explicitly.

Formulas \eqref{PNL1}, \eqref{Pnui}, \eqref{SZX} show that $P_N(\{\lambda\})$ is uniquely
determined by its properties (1)\,--\,(4). Therefore, $P_N(\{\lambda\})=Z_N (\{\lambda\})$
by Proposition \ref{prop:props-rat}.
\end{proof}

\subsection{Determinant representation}\label{sec42}
Consider polynomials $p_1 (\lambda), \dots, p_N (\lambda)$ in one variable
of degree at most $N-1$,
\begin{equation}
p_i(\lambda)=\sum_{j=1}^N p_{ij}\lambda^{j-1},\qquad i=1,\ldots,N.
\end{equation}
Set
\begin{equation}
Q_N(\{p\})=\det[p_{ij}]_{i,j=1,\ldots,N},
\end{equation}
where $\{p\}=(p_1,\ldots,p_N)$.
The polynomials $p_1(\lambda),\dots,p_N(\lambda)$ form a basis of polynomials
in $\lambda$ of degree $N-1$ if and only if $Q_N(\{p\})\ne 0$. Notice the equality
\begin{equation}\label{detp=QV}
\det\left[p_k(\lambda_j)\right]_{j,k=1,\ldots,N}=
Q_N(\{p\})\prod_{1\leq j<k\leq N}(\lambda_k-\lambda_j),
\end{equation}
that follows from reading the formula
\begin{equation}
p_k(\lambda_j)=\sum_{i=1}^N p_{ki}\lambda_j^{i-1}
\end{equation}
as a product of matrices
\begin{equation}
\left[p_k(\lambda_j)\right]_{j,k=1,\ldots, N}=
[p_{ki}]_{i,k=1,\ldots, N}\;[\lambda_j^{i-1}]_{i,j=1,\ldots, N},
\end{equation}
and the Vandermonde determinant formula.

Our main result in the rational case can be formulated as follows.

\begin{theorem}\label{theorem:det-rat}
The partition function of the six-vertex model with DWBC
with the rational weights \eqref{eq:weights-rat} can be written in the form
\begin{equation}\label{eq:det-rat}
Z_N(\{\lambda\})=\frac{
\det
\left[
p_k (\lambda_j) a (\lambda_j) -
p_k (\lambda_j + 1) d (\lambda_j)
\right]_{j, k=1,\ldots, N}}
{\det\left[p_k(\lambda_j)\right]_{j,k=1,\ldots,N}}
\end{equation}
where $p_1 (\lambda), \dots, p_N (\lambda)$ form a basis of polynomials in $\lambda$
of degree $N-1$, and $a (\lambda)$ and $d (\lambda)$ are the eigenvalues
of the operators $A (\lambda)$ and $D (\lambda)$:
\begin{equation}\label{ad}
a(\lambda)=\prod_{k=1}^{N}(\lambda-\nu_k+1),\qquad
d(\lambda)=\prod_{k=1}^{N}(\lambda-\nu_k).
\end{equation}
Notice formula \eqref{detp=QV} for the denominator in \eqref{eq:det-rat}.
\end{theorem}

\begin{proof}
First observe that the denominator in \eqref{eq:det-rat} is a nonzero polynomial
since $p_1 (\lambda), \dots, p_N (\lambda)$ form a basis of polynomials in $\lambda$
of degree $N-1$, see formula \eqref{detp=QV}.

Denote by $P_N(\{\lambda\})$ the right-hand side of formula \eqref{eq:det-rat}.
The numerator in \eqref{eq:det-rat} is an antisymmetric polynomial in $\lambda_1,\allowbreak\ldots,\allowbreak\lambda_N$,
and hence, is divisible by the denominator. Therefore,
$P_N(\{\lambda\})$ is a symmetric polynomial in $\lambda_1,\allowbreak\ldots,\allowbreak\lambda_N$.
To prove the theorem, we will show that $P_N(\{\lambda\})$ obeys also properties (2)\,--\,(4)
described in Proposition \ref{theorem:dim-rat}. Then $P_N(\{\lambda\})=Z_N(\{\lambda\})$
by Proposition \ref{theorem:dim-rat}.

For the second property, observe that both $a(\lambda)$ and $d(\lambda)$ are monic polynomials
of degree $N$. Thus for each $k=1,\ldots,N$, the polynomial
$p_k (\lambda) a (\lambda) - p_k (\lambda + 1) d (\lambda)$ has degree at most $2N-2$.
Therefore, the degree of $P_N(\{\lambda\})$ in each of the variables $\lambda_1,\ldots,\lambda_N$
does not exceed $N-1$.

The third property follows from the equalities
\begin{equation}\label{adzero}
a (\nu_j-1)=d(\nu_j)=0, \qquad j = 1, \dots, N,
\end{equation}
so that after the substitution $\lambda_1=\nu_j$, $\lambda_2=\nu_j-1$,
the first two rows of the matrix
$\left[p_k (\lambda_j)a(\lambda_j)-p_k (\lambda_j+1)d(\lambda_j)\right]_{j,k=1,\ldots,N}$
become proportional,
\begin{equation}
\begin{pmatrix}
p_1 (\nu_j) a (\nu_j) & p_2 (\nu_j) a (\nu_j) & \dots & p_N (\nu_j) a (\nu_j)
\\
p_1 (\nu_j) d (\nu_j-1) & p_2 (\nu_j) d (\nu_j-1) & \dots & p_N (\nu_j) d (\nu_j-1)
\\
\hdotsfor{4}
\\
\hdotsfor{4}
\end{pmatrix},
\end{equation}
and its determinant vanishes. Hence,
$P_N(\nu_j, \nu_j - 1,\lambda_3, \dots,\allowbreak \lambda_N) =0$.

To verify the fourth property, we take $\{\lambda\}=\{\nu\}$.
Taking into account the second equality in \eqref{adzero}, we get
\begin{multline}
\det\left[p_k (\nu_j) a (\nu_j) -p_k (\nu_j + 1) d (\nu_j)\right]_{j, k=1,\ldots, N}
\\ 
=\det\left[p_k (\nu_j) a (\nu_j)\right]_{j, k=1,\ldots, N}
=\left(\prod_{j=1}^{N} a(\nu_j)\right)\det\left[p_k (\nu_j) \right]_{j, k=1,\ldots, N}.
\end{multline}
Therefore by the first of formulas \eqref{ad},
\begin{equation}
P_N(\{\nu\})=\prod_{j,k = 1}^N (\nu_j - \nu_k + 1),
\end{equation}
that completes the proof of Theorem \ref{theorem:det-rat}.
\end{proof}

Let us now turn to examples. The first example
to consider is
the celebrated determinant formula \eqref{eq:Z-IK}.
To
get it, we choose $p_1(\lambda),\ldots\,p_N(\lambda)$ to be the
interpolating polynomials of degree $N-1$ for the points $\nu_1,\ldots,\nu_N$,
\begin{equation}
p_k (\lambda)
= \prod_{\substack{l = 1\\l \neq k}}^N (\lambda - \nu_l),\qquad k=1,\ldots,N.
\end{equation}
Since $p_k(\lambda)=d(\lambda)/(\lambda-\nu_k)$ and $p_k(\lambda+1)=a(\lambda)/(\lambda-\nu_k+1)$,
we have
\begin{multline}
p_k (\lambda_j) a (\lambda_j) - p_k (\lambda_j + 1) d (\lambda_j)
\\ 
=\frac{d(\lambda_j)a (\lambda_j)}{\lambda_j-\nu_k}-
\frac{a(\lambda_j) d (\lambda_j)}{\lambda_j-\nu_k+1}
=\frac{a(\lambda_j) d (\lambda_j)}{(\lambda_j-\nu_k)(\lambda_j-\nu_k+1)},
\end{multline}
and for the numerator in formula \eqref{eq:det-rat}, we obtain
\begin{multline}\label{IKnum}
\det
\left[p_k (\nu_j) a (\nu_j) -p_k (\nu_j + 1) d (\nu_j)\right]_{j, k=1,\ldots, N}
\\
=\left(\prod_{j=1}^N a(\lambda_j)d(\lambda_j) \right)
\det
\left[\frac{1}{(\lambda_j-\nu_k)(\lambda_j-\nu_k+1)}\right]_{j, k=1,\ldots, N}.
\end{multline}
For the denominator in formula \eqref{eq:det-rat}, we use first formula \eqref{detp=QV}
and evaluate the factor $Q_N(\{p\})$ employing the same formula with $\{\lambda\}=\{\nu\}$,
\begin{multline}\label{IKden}
Q_N(\{p\})=\frac{\det\left[p_k(\nu_j)\right]_{j,k=1,\ldots,N}}{\prod_{1\leq j<k\leq N}(\nu_k-\nu_j)}
=\frac{\det\left[\delta_{kj}\prod_{i=1,i\ne k}^N(\nu_k-\nu_i)\right]_{j,k=1,\ldots,N}}
{\prod_{1\leq j<k\leq N}(\nu_k-\nu_j)}
\\
=\prod_{1\leq j<k\leq N}(\nu_j-\nu_k).
\end{multline}
Taking together formulas \eqref{Vander-rat}, \eqref{eq:det-rat}, \eqref{IKnum}, and \eqref{IKden},
we get formula \eqref{eq:Z-IK} in the rational case.

The second example we would like to mention is the choice of $p_1(\lambda),\ldots\,p_N(\lambda)$
as monomials,
\begin{equation}
p_k (\lambda) = \lambda^{k - 1},\qquad k=1,\ldots,N,
\end{equation}
In this case, representation \eqref{eq:det-rat} clearly reproduces Kostov's
determinant formula \eqref{eq:Z-K}.

\section{Determinant representation in the trigonometric case}
\label{sec:5}

In this section, we
generalize the
results of the previous section
to the trigonometric case.
As before,
our main result is
the determinant representation
for the partition function
that depends on an arbitrary basis of polynomials.
Particular
examples are
given by
representations \eqref{eq:Z-IK} and \eqref{eq:Z-K-trig}.

\subsection{Properties of the partition function}

As we have
observed
studying the rational case, an important role
is played by a polynomial dependence of the partition function on the inhomogeneity parameters.
The polynomial dependence in the trigonometric case can be introduced by
the ``change of variables''
\begin{equation}
x_j=q^{2\lambda_j},\qquad
y_k=q^{2\nu_k},\qquad
q=\rme^{\rmi\gamma},
\end{equation}
so the Boltzmann weights \eqref{eq:weights-trig} in the trigonometric case are
\begin{equation}\label{abc-xyq}
a(\lambda_j,\nu_k)=\frac{q x_j-y_k q^{-1}}{2\rmi\,(x_jy_k)^{1/2}},\quad
b(\lambda_j,\nu_k)=\frac{x_j-y_k}{2\rmi\,(x_jy_k)^{1/2}},\quad
c(\lambda_j,\nu_k)=\frac{q-q^{-1}}{2\rmi}.
\end{equation}
Then the partition function of the six-vertex model with DWBC
in the trigonometric case has the form
\begin{equation}\label{ZwtZ}
Z_N (\{\lambda\}; \{\nu\})=
\frac{\wt Z_N (\{x\}; \{y\})}{(2\rmi)^{N^2}\prod_{j = 1}^N(x_jy_j)^{(N-1)/2}},
\end{equation}
where $\wt Z_N(\{x\}; \{y\})$ is a polynomial in $\{x\}\equiv x_1,\ldots,x_N$ and
$\{y\}\equiv y_1,\ldots,y_N$.

Indeed, for every allowed configuration of the six-vertex with DWBC, each horizontal and
each vertical line contains an odd number of arrow reversing vertices (the $5$th and
$6$th types, see Fig.~\ref{fig:six-vertices}), that is, of the $c$-weight vertices. Therefore,
the weight of every configuration is a polynomial of the variables $x_1,\ldots,x_N$,
$y_1,\ldots,y_N$, divided by the product $\prod_{j = 1}^N(x_jy_j)^{(N-1)/2}$.

We apply the method developed in the rational case to the function $\wt Z_N$.
We begin with the counterpart of Proposition \ref{prop:props-rat}.

\begin{proposition}\label{prop:props-trig}
The partition function $\wt Z_N (\{x\})$ as a function of the variables $x_1,\dots,x_N$
has the following properties:
\begin{enumerate}
\item
It is a symmetric polynomial in $x_1,\dots,x_N$;
\item
The degree in each of the variables $x_1,\dots,x_N$ equals $N - 1$;
\item
For each pair of variables, say $x_1$, $x_2$, one has
\begin{equation}\label{eqyy}
\wt Z_N(y_j, q^{-2} y_j , x_3, \dots, x_N) = 0, \qquad j = 1, \dots, N;
\end{equation}
\item
As $\{x\}=\{y\}$, the following holds:
\begin{equation}\label{Zyy}
\wt Z_N(\{y\}) = (q-q^{-1})^N\prod_{\substack{j,k=1\\j\neq k}}^N (qy_j-q^{-1}y_k).
\end{equation}
\end{enumerate}
\end{proposition}

\begin{proof}
The proof is completely similar to the rational case, see Proposition \ref{prop:props-rat}.
\end{proof}

Now we show that the system of equations \eqref{eqyy}, \eqref{Zyy} has a unique solution.

\begin{proposition}\label{theorem:dim-trig}
Let $P_N (\{x\})$ be a polynomial in the variables $x_1,\dots,x_N$,
depending on the parameters $\{y\}=y_1,\ldots,y_N$, that has the following properties:
\begin{enumerate}
\item
It is a symmetric polynomial in $x_1,\dots,x_N$;
\item
The degree in each of the variables $x_1,\dots,x_N$ equals $N - 1$;
\item
For each pair of variables, say $x_1,x_2$, one has
\begin{equation}\label{Peqyy}
P_N(y_j, q^{-2} y_j , x_3, \dots, x_N) = 0, \qquad j = 1, \dots, N;
\end{equation}
\item
As $\{x\}=\{y\}$, the following holds:
\begin{equation}\label{Pyy}
P_N(\{y\}) = (q-q^{-1})^N\prod_{\substack{j,k=1\\j\neq k}}^N (qy_j-q^{-1}y_k).
\end{equation}
\end{enumerate}
Then, $P_N (\{x\})=\wt Z_N (\{x\})$.
\end{proposition}
\begin{proof}
The proof is the same as that of Proposition \ref{theorem:dim-rat}.
We reproduce it in the condensed form for clarity.

The proof is by induction on $N$. The base of induction at $N=1$ is clear.
For the induction step, we expand the polynomial $P_N (\{x\})$ in the variable $x_1$
\begin{equation}\label{PxL1}
P_N (\{x\})= \sum_{i = 1}^N P_N (y_i, x_2, \dots, x_N)
\prod_{\substack{j = 1\\ j \neq i}}^{N}\frac{x_1 - y_j }{y_i - y_j},
\end{equation}
and use \eqref{Peqyy} to get
\begin{align}\label{Pyi}
P_N (y_i, x_2, \dots, x_N) &{}= S_{N-1,i} (x_2, \dots, x_N)
\notag
\\
& {}\times (q-q^{-1})\,
\prod_{k = 2}^N (qx_k - q^{-1}y_i)\,\prod_{\substack{j=1\\j\neq i}}^N (qy_i - q^{-1}y_j),
\end{align}
$i=1,\ldots,N$, for some polynomials $S_{N-1,i} (x_2, \dots, x_N)$ in $x_2, \dots,x_N$.

For each $i=1,\ldots,N$, the polynomial $S_{N-1,i} (x_2, \dots,x_N)$ has the properties:
\begin{enumerate}
\item
It is symmetric in $x_2, \dots, x_N$;
\item
The degree in each of the variables $x_2, \dots, x_N$ equals $N - 2$;
\item
For each pair of the variables, say $x_2,x_3$,
one has
\begin{equation}\label{Seqyy}
S_{N-1,i}(y_j, q^{-2}y_j,x_4, \dots, x_N) =0,\qquad j = 1, \dots, i-1, i+1,\ldots, N;
\end{equation}
\item
The following holds:
\begin{equation}\label{Syy}
S_{N-1,i}(\{y\}\setminus y_i) = (q-q^{-1})^{N-1}
\prod_{\substack{j,k=1\\j,k\ne i,\,j\neq k}}^N (qy_j-q^{-1}y_k),
\end{equation}
where we use the notation
$\,\{y\}\setminus y_i:=y_1,\ldots,y_{i-1},y_{i+1},\ldots,y_N$.
\end{enumerate}
Hence by the induction assumption, for each $i=1,\ldots,N$, one has
\begin{equation}\label{SZXy}
S_{N-1,i} (x_2, \dots, x_N) = \wt Z_{N-1} (x_2, \dots, x_N;\{y\}\setminus y_i),
\end{equation}
where we indicate in the right-hand side the dependence on the parameters
$\{y\}$ explicitly.

Formulas \eqref{PxL1}, \eqref{Pyi}, \eqref{SZXy} show that $P_N(\{x\})$ is uniquely
determined by its properties (1)\,--\,(4). Therefore, $P_N(\{x\})=\wt Z_N (\{x\})$
by Proposition \ref{prop:props-trig}.
\end{proof}

\subsection{Determinant representation}\label{sec52}

Our main result in the trigonometric case can be formulated as follows.

\begin{theorem}\label{theorem:det-trig}
The partition function of the six-vertex model with DWBC
with the trigonometric weights \eqref{eq:weights-trig} can be written in the form
\begin{equation}\label{det-trig}
\wt Z_N(\{x\})=\frac{\det\left[
p_k (x_j) \tilde a (x_j)-q^{-N+2} p_k (q^2x_j)\tilde d (x_j)\right]_{j,k=1,\ldots,N}}
{\left(\prod_{k=1}^Ny_k\right)\det[p_k (x_j)]_{j,k=1,\ldots,N}},
\end{equation}
and in the form
\begin{equation}\label{det-trig2}
\wt Z_N(\{x\})=\frac{\det\left[
p_k (x_j) \tilde a (x_j)-q^{-N} p_k (q^2x_j) \tilde d (x_j)\right]_{j, k=1,\ldots, N}}
{\left(\prod_{j=1}^Nx_j\right)\det[p_k (x_j)]_{j,k=1,\ldots,N}}.
\end{equation}
Here $p_1 (x), \dots, p_N (x)$ form a basis of polynomials in $x$ of degree $N-1$,
and $\tilde a(x)$, $\tilde d(x)$ are
\begin{equation}
\tilde a(x)=\prod_{k=1}^{N}\left(q x-q^{-1}y_k\right),\qquad
\tilde d(x)=\prod_{k=1}^{N}\left(x-y_k\right).
\end{equation}
Notice formula \eqref{detp=QV},
\begin{equation}
\det[p_k (x_j)]_{j,k=1,\ldots,N}=Q_N(\{p\})\prod_{1\le j<k\le N}(x_k-x_j).
\end{equation}
\end{theorem}
\begin{proof}
The proof is similar to the rational case. We verify that the right-hand sides of
formulas \eqref {det-trig} and \eqref{det-trig2} obeys properties (1)\,--\,(4) stated
in Proposition~\ref{prop:props-trig}. This is done along the same lines as in the proof
of Theorem \ref{theorem:det-rat}.
\end{proof}

Note that the functions $\tilde a(x)$ and $\tilde d(x)$ are related to vacuum eigenvalues
of the $A$- and $D$-operators,
\begin{equation}\label{aadd}
a(\lambda)=\frac{\tilde a(x)}{(2\rmi)^N x^{N/2}\prod_{j=1}^Ny_j^{1/2}},\qquad
d(\lambda)=\frac{\tilde d(x)}{(2\rmi)^N x^{N/2}\prod_{j=1}^Ny_j^{1/2}},\qquad
x=q^{2\lambda}.
\end{equation}

Let us now turn to examples. The first example
to consider is determinant formula \eqref{eq:Z-IK} in the trigonometric case.
To get it, we choose $p_1(x),\ldots\,p_N(x)$ to be the interpolating polynomials
of degree $N-1$ for the points $y_1,\ldots,y_N$,
\begin{equation}
p_k (x)
= \prod_{\substack{l = 1\\l \neq k}}^N\,(x-y_l),\qquad k=1,\ldots,N.
\end{equation}
Since $p_k(x)=\tilde d(x)/(x-y_k)$, $\,p_k(q^2x)=q^{N-1}\tilde a(x)/(qx-q^{-1}y_k)$, we have
\begin{multline}
p_k (x_j)\tilde a (x_j)- q^{-N+2}p_k (q^2x_j)\tilde d (x_j)
\\
=\frac{\tilde d(x_j)\tilde a (x_j)}{x_j-y_k}-
\frac{q\tilde a(x_j)\tilde d (x_j)}{qx_j-q^{-1}y_k}
=\frac{(q-q^{-1})\,y_k\,\tilde a(x_j)\tilde d (x_j)}{(x_j-y_k)\left(qx_j-q^{-1}y_k\right)}
\end{multline}
and
\begin{multline}
p_k (x_j)\tilde a(x_j) - q^{-N}p_k (q^2x_j)\tilde d (x_j)
\\
=\frac{\tilde d(x_j)\tilde a (x_j)}{x_j-y_k}-
\frac{q^{-1}\tilde a(x_j)\tilde d (x_j)}{qx_j-q^{-1}y_k}
=\frac{(q-q^{-1})\,x_j\,\tilde a(x_j)\tilde d (x_j)}{(x_j-y_k)\left(qx_j-q^{-1}y_k\right)}.
\end{multline}
Using these formulas to transform the numerators in formulas \eqref{det-trig},
\eqref{det-trig2}, and formulas \eqref{detp=QV}, \eqref{IKden} to deal with the denominators,
we get in both cases
\begin{equation}\label{wtZfin}
\wt Z_N(\{x\})=
\frac{\prod_{j=1}^N \tilde a(x_j)\tilde d(x_j)}{\prod_{1\leq j<k\leq n}(x_k-x_j)(y_j-y_k)}
\det\left[\frac{q-q^{-1}}{(x_j-y_k)(qx_j-q^{-1}y_k)}\right]_{j,k=1,\ldots,N}.
\end{equation}
Taking into account formulas \eqref{Vander-trig}, \eqref{eq:a-and-d}, \eqref{abc-xyq}, 
and \eqref{ZwtZ}, the last equality
yields formula \eqref{eq:Z-IK} in the trigonometric case.

The second example is the choice of $p_1(x),\ldots,p_N(x)$ as monomials,
\begin{equation}
p_k (x) = x^{k - 1},\qquad k=1,\ldots,N,
\end{equation}
In this case, representation \eqref{det-trig} gives formula \eqref{eq:Z-K-trig}
and representation \eqref{det-trig2} gives formula \eqref{eq:Z-K-trig2}.


\section{Acknowledgments}

The first and second authors (M.D.M. and A.G.P.) are supported in part by the
BASIS Foundation, grant \#21-7-1-32-3. The first author (M.D.M.)
also acknowledges partial support from the
Leonhard Euler International Mathematical Institute,
agreement \#075--15--2019--1620.

\bibliography{mpt_bib}
\end{document}

%% file: fig-six-vertices.tex
\begin{tikzpicture}[scale = 0.7, thick]
	\draw [semithick] (0, 1) -- ++(2, 0);
	\draw [semithick] (1, 0) -- ++(0, 2);
	\draw [->-] (0, 1) -- ++(1, 0);
	\draw [->-] (1, 1) -- ++(1, 0);
	\draw [->-] (1, 0) -- ++(0, 1);
	\draw [->-] (1, 1) -- ++(0, 1);
	\node at (1, -1) {$w_1$};
\end{tikzpicture}
\quad
\begin{tikzpicture}[scale = 0.7, thick]
	\draw [semithick] (0, 1) -- ++(2, 0);
	\draw [semithick] (1, 0) -- ++(0, 2);
	\draw [->-] (1, 1) -- ++(-1, 0);
	\draw [->-] (2, 1) -- ++(-1, 0);
	\draw [->-] (1, 1) -- ++(0, -1);
	\draw [->-] (1, 2) -- ++(0, -1);
	\node at (1, -1) {$w_2$};
\end{tikzpicture}
\quad
\begin{tikzpicture}[scale = 0.7, thick]
	\draw [semithick] (0, 1) -- ++(2, 0);
	\draw [semithick] (1, 0) -- ++(0, 2);
	\draw [->-] (0, 1) -- ++(1, 0);
	\draw [->-] (1, 1) -- ++(1, 0);
	\draw [->-] (1, 1) -- ++(0, -1);
	\draw [->-] (1, 2) -- ++(0, -1);
	\node at (1, -1) {$w_3$};
\end{tikzpicture}
\quad
\begin{tikzpicture}[scale = 0.7, thick]
	\draw [semithick] (0, 1) -- ++(2, 0);
	\draw [semithick] (1, 0) -- ++(0, 2);
	\draw [->-] (1, 1) -- ++(-1, 0);
	\draw [->-] (2, 1) -- ++(-1, 0);
	\draw [->-] (1, 0) -- ++(0, 1);
	\draw [->-] (1, 1) -- ++(0, 1);
	\node at (1, -1) {$w_4$};
\end{tikzpicture}
\quad
\begin{tikzpicture}[scale = 0.7, thick]
	\draw [semithick] (0, 1) -- ++(2, 0);
	\draw [semithick] (1, 0) -- ++(0, 2);
	\draw [->-] (0, 1) -- ++(1, 0);
	\draw [->-] (2, 1) -- ++(-1, 0);
	\draw [->-] (1, 1) -- ++(0, -1);
	\draw [->-] (1, 1) -- ++(0, 1);
	\node at (1, -1) {$w_5$};
\end{tikzpicture}
\quad
\begin{tikzpicture}[scale = 0.7, thick]
	\draw [semithick] (0, 1) -- ++(2, 0);
	\draw [semithick] (1, 0) -- ++(0, 2);
	\draw [->-] (1, 1) -- ++(-1, 0);
	\draw [->-] (1, 1) -- ++(1, 0);
	\draw [->-] (1, 0) -- ++(0, 1);
	\draw [->-] (1, 2) -- ++(0, -1);
	\node at (1, -1) {$w_6$};
\end{tikzpicture}

%% file: fig-DWBC.tex
\begin{tikzpicture}[scale = 0.65, thick]
	\def\n{6}
	\def\d{0.2}	
	\foreach \x in {1, ..., \n}{
		\draw (\x, 0) -- ++(0, \n + 1);
		\draw[->-] (\x, 0) -- ++(0, 1);
		\draw[->-] (\x, \n + 1) -- ++(0, -1);
	}
	\foreach \y in {1, ..., \n}{
		\draw (0, \y) -- ++(\n + 1, 0);
		\draw[->-] (\n, \y) -- ++(1, 0);
		\draw[->-] (1, \y) -- ++(-1, 0);
	}
	\foreach \x in {1, 2}{
		\node at (\n + 1 - \x, 0) [below] {$\lambda_\x$};
		\node at (\n + 1, \n + 1 - \x) [right] {$\nu_\x$};
	}
	\node at (1, 0) [below] {$\lambda_N$};
	\node at (\n + 1, 1) [right] {$\nu_N$};
\end{tikzpicture}